\documentclass[runningheads]{llncs}
\usepackage[T1]{fontenc}
\usepackage{graphicx}

\usepackage{hyperref}

\usepackage{amsmath}
\usepackage{cleveref}
\usepackage{booktabs}
\usepackage{xcolor}
\usepackage{color, colortbl}

\usepackage{tikz}
\usetikzlibrary{tikzmark}

\usepackage{amssymb}
\usepackage{comment}
\usepackage{enumitem}
\usepackage{multicol}

\newcommand{\overbar}[1]{\mkern 3mu\overline{\mkern-3mu#1\mkern-3mu}\mkern 3mu}

\newcommand{\vdashx}[1]{\mathrel{\vdash\mkern-7mu_{\raisebox{-0.3ex}{\scriptsize $#1$}}}}

\begin{document}
\title{Inclusion with repetitions and Boolean constants -- implication problems revisited}
\titlerunning{Inclusion with repetitions and Boolean constants}
%
\author{Matilda Häggblom\orcidID{0009-0003-6289-7853} }
\authorrunning{M. Häggblom}
%
\institute{University of Helsinki, Finland \\
\email{}}
\maketitle              
\begin{abstract}
Inclusion dependencies form one of the most widely used dependency classes. We extend existing results on the axiomatization and computational complexity of their implication problem to two extended variants. We present an alternative completeness proof for standard inclusion dependencies and extend it to inclusion dependencies with repetitions that can express equalities between attributes.  The proof uses only two values, enabling us to work in the Boolean setting.
Furthermore, we study inclusion dependencies with Boolean constants, provide a complete axiomatization and show that no such system is $k$-ary. Additionally, the decision problems for both extended versions remain PSPACE-complete. The extended inclusion dependencies examined are common in team semantics, which serves as the formal framework for the results.

\end{abstract}

\section{Introduction}

Inclusion dependencies have been extensively studied in database theory due to their widespread use. A fundamental question for dependencies is the implication problem: Does a set of dependencies imply a given dependency? The axiomatization and computational complexity of the implication problem for standard inclusion dependencies were established in \cite{CASANOVA198429}. In this work, we identify and address gaps in the literature concerning the implication problem for natural variants of inclusion dependencies, namely inclusion in the Boolean setting, inclusion with repetitions, and inclusion with Boolean constants.

The extended variants of the inclusion dependencies are common in the team semantic literature, making it a convenient setting for presenting the results. Team semantics was introduced in \cite{hodges1997,hodges19972} and further developed in \cite{vaananen2007}. In particular, the inclusion dependencies were adapted to the team semantic setting in \cite{GALLIANI201268}, where teams correspond to uni-relation databases, variables to attributes, and atoms to dependencies. Consequently, the results of this paper are immediately transferable between the database and team semantic setting.

An inclusion atom $x\subseteq y$, where $x$ and $y$ are finite sequences of variables of equal length, is satisfied in a team if the values of $x$ are among those of $y$. In database theory, it is common to prohibit repetitions of variables within the variable sequences $x$ and $y$. Removing this restriction, as often done in team semantics, enables inclusion atoms to express equalities such as $x=y$ via $xy\subseteq xx$. Examining this class of atoms is thus a natural consideration. 

Furthermore, we generalize the syntax in the Boolean setting by allowing the constants $\top$ and $\bot$ within the sequences. This enables us to specify the inclusion of precise values. For instance, we express that a proposition $p$ must be true somewhere by $\top\subseteq p$. Additionally, contradictions like $\top\subseteq \bot$ can be expressed, which is not possible with standard inclusion atoms. 

To illustrate the difference in expressive power between the three classes of inclusion atoms, we provide simple natural language examples in the Boolean setting, where variables take only the values $0$ and $1$,  
Consider the team $T=\{s_{2021},s_{2022},s_{2023}\}$, and let $\{p_a,p_b\}$ be a set of two propositional symbols, for which $s_{2021}(p_a)=1$ if and only if student $a$ passed their exams in $2021$, etc.
\begin{enumerate}[align=left]
\item[\textbf{Without repetitions:}] $T\models p_a\subseteq p_b$ expresses that during the three years, student $a$ failing one year means that also student $b$ failed at least once. Conversely, if student $a$ passed one year, then student $b$ must have passed at least once.
    
    \item[\textbf{With repetitions:}] $T\models p_ap_b\subseteq p_ap_a$ expresses that each year, student $b$ passed their exams if and only if student $a$ did. Note that this is not implied by $T\models p_a\subseteq p_b$ and $T\models p_b\subseteq p_a$, since the students' similar results can be witnessed at different years.

     \item[\textbf{With Boolean constants:}] 
$T\models p_a\bot\subseteq \top p_b$ expresses that student $a$ passed their exams every year, while student $b$ failed at least once. 
\end{enumerate}

Although the inclusion atoms we consider have been studied within the larger language of propositional inclusion logic \cite{yang_propositional_2022}, their axiomatization from \cite{CASANOVA198429} has previously not been confirmed complete in the Boolean setting. We address this gap in the literature by providing an alternative completeness proof using only two values. Indeed, the systems for repetition-free inclusion atoms remains complete. This result is not to be taken for granted, since, for instance, any complete system for exclusion atoms in the Boolean setting necessarily extends their usual system from \cite{Casanova1983TowardsAS}, as noted in \cite{haggblom2024}. 

Without a restriction on the number of possible values, inclusion atoms have Armstrong relations \cite{inclARMST}, i.e., for any set of inclusion atoms there exists a team that satisfies exactly the atoms entailed by the set. We prove that this is not the case in the Boolean setting, showing that despite the proof systems being the same, some basic semantic properties differ.

Inclusion with repetitions has been considered in \cite{mitchell} together with functional dependencies. We extract the inclusion rules from their system and prove completeness. Furthermore, we confirm that the implication problem for finite sets of atoms with repetitions is PSPACE-complete, as in the repetition-free case \cite{CASANOVA198429}.

Team-based propositional logic can be extended with dependency atoms to obtain propositional dependence logic and propositional inclusion logic. Dependence atoms (based on functional dependencies) are downward-closed, meaning that if a team satisfies an atom, so do its subteams. Adding their Boolean versions to propositional logic forms an expressively complete logic for all downward-closed team properties \cite{MR3488885}. However, for propositional inclusion logic, the situation is different. Inclusion atoms are not downward-closed, but union-closed, i.e., if multiple teams satisfy the same atom, so does their union. Adding only their Boolean variants to propositional logic does not yield an expressively complete logic for all union-closed team properties (with the empty team), as noted in \cite{yang_propositional_2022}: over one variable, $p\subseteq p$ is a tautology and does not extend the expressivity. Instead, to achieve an expressively complete logic, we require inclusion atoms with Boolean constants. Note that the axiomatization of propositional inclusion logic in \cite{yang_propositional_2022} does not give the rules for inclusion atoms with Boolean constants simpliciter, since the rules in \cite{yang_propositional_2022} make use of the connectives of the language. 

We introduce a complete axiomatization of inclusion with Boolean constants and show that no such system can be $k$-ary. A $k$-ary system requires at most $k$ assumptions in any of its rules. This result contrasts with the complete systems for the other inclusion atoms considered, which are $2$-ary.
Still, for this extended class, the decision problem for finite sets of atoms is PSPACE-complete. 

The paper is organized as follows: We recall the basic definitions and results of inclusion atoms in \Cref{sec prel}, and their repetition-free version in \Cref{sec rep free}. We examine the implication problem for (Boolean) inclusion atoms with repetitions in \Cref{sec inc}, and establish corresponding results for inclusion atoms with Boolean constants in \Cref{sec bool}. We conclude with a summary and directions for future research in \Cref{sec conc}.

\section{Preliminaries}\label{sec prel}

We recall the basic definitions for teams and inclusion atoms, see, e.g., \cite{GALLIANI201268} for a reference. 

A team $T$ is a set of assignments $s:\mathcal{V}\longrightarrow M$, where $\mathcal{V}$ is a set of variables and $M$ is a set of values. We write $x_i,y_i,\dots$ for individual variables and $x,y,\dots$ for finite (possibly empty) sequences of variables. Given a sequence $x$, we denote its i:th variable by $x_i$. Let $x= x_1\dots x_n$ and $y= y_1\dots y_m$. We write $s(x)$ as shorthand for the tuple $s(x_1)\dots s(x_n)$. The concatenation $xy$ is the sequence $ x_1\dots x_n y_1\dots y_m$. The sequences $x$ and $y$ are equal if and only if they are of the same length, denoted $|x|=|y|$, and $x_1=y_1,\dots, x_n=y_n$. We let $V(x)$ denote the set of variables appearing in the sequence $x$. 

Inclusion atoms are of the form $x\subseteq y$, where $|x|=|y|$ is the arity of the atom. We often refer to the variables in $x$ as the variables on the left-hand side, etc. We recall the semantics of inclusion atoms, 
\begin{equation*}
     T\models x\subseteq y \text{ if and only if for all } s\in T \text{ there exists } s^\prime\in T : s(x)=s^\prime(y). 
\end{equation*}

Let $T[x]:=\{s(x)\mid s\in T\}$. The definition of $T\models x\subseteq y$ expresses that any value of $x$ in $T$ is a value of $y$ in $T$, giving us the equivalent semantic clause: $$T\models x\subseteq y \text{ if and only if }T[x]\subseteq T[y].$$

Inclusion atoms are union-closed: if $T_i\models x\subseteq y$ for each $i$ in a nonempty index set $I$, then $\bigcup_{I\in I}T_i\models x\subseteq y$. They also have the empty team property: all inclusion atoms are satisfied by the empty team. We call inclusion atoms of the form $x\subseteq x$ \emph{trivial}, since they are satisfied in all teams. 

We write $\Sigma\models u\subseteq v$ if for all teams $T$, whenever $T\models x\subseteq y$ for all $x\subseteq y\in \Sigma$, $T\models u\subseteq v$. If $\Sigma=\{x\subseteq y\}$ is a singleton, we drop the brackets and write $x\subseteq y\models u\subseteq v$. 


For inclusion atoms in the Boolean setting, we denote variable sequences by $p,q,\dots$ and restrict the values of individual variables $p_i,q_i,\dots$ to $0$ and $1$. The semantic clause for $T\models p\subseteq q$ and the above-mentioned union closure and empty team properties are as in the general case.

\section{Repetition-free inclusion} \label{sec rep free}

We say that an inclusion atom $x\subseteq y$ is \emph{repetition-free} if there are no repeated variables within the sequences $x$ and $y$, respectively. A repetition-free inclusion atom is $x_1\subseteq x_1$, while $x_1x_1\subseteq y_1y_2$ is not. In the database literature, inclusion dependencies are often defined to be repetition-free. 

We recall the complete system for repetition-free inclusion atoms.

\begin{definition}
\label{Inclusion usual rules}
    [\cite{CASANOVA198429}]
    The rules for repetition-free inclusion atoms are $I1$-$I4$. 
\begin{enumerate}[align=left]
\item[$(I1)$] $x\subseteq x$.

\item[$(I2)$] If $x\subseteq z$ and $z\subseteq y$, then $x\subseteq y$.

\item[$(I3)$] If $xyz\subseteq uvw$, then $xzy\subseteq uwv$.$^*$ \hfill $^*$$|x|=|u|$ and $|y|=|v|$

\item[$(I4)$] If $xy\subseteq uv$, then $x\subseteq u$.
\end{enumerate}
\end{definition} 

For a set of rules $R$, we write $\Sigma\vdashx{R} x\subseteq y$ when the rules we apply in the derivation are restricted to the ones in $R$. If $R$ is a singleton, we omit the brackets, and if $R=\{I1-I4\}$, we write $\Sigma\vdash x\subseteq y$. We follow this convention for all types of inclusion atoms and their respective proof systems.

\begin{theorem}[Completeness \cite{CASANOVA198429}]
    Let $\Sigma\cup\{x\subseteq y\}$ be a set of repetion free inclusion atoms. $\Sigma\models x\subseteq y$ if and only if $\Sigma \vdash x\subseteq y$.
\end{theorem}

 In \Cref{complincl} of the next section, we extend the completeness result to the Boolean setting.

\section{Inclusion with repetitions} \label{sec inc}

We remove the restriction on repetition of variables and consider \emph{inclusion atoms}, common in the team semantics literature. They are syntactically and semantically extended versions of their repetition-free counterparts.

The rules for inclusion atoms are presented in \cite{mitchell} as part of a larger system for inclusion and functional dependencies. We extract the rules for inclusion atoms and show that the system is complete also in the Boolean setting.

\begin{definition} \label{Inclusion full rules}
    [\cite{mitchell}]
The rules for inclusion atoms are $I1$-$I6$.
\begin{enumerate}[align=left]
\item[$(I1)$] $x\subseteq x$.

\item[$(I2)$] If $x\subseteq z$ and $z\subseteq y$, then $x\subseteq y$.

\item[$(I3)$] If $xyz\subseteq uvw$, then $xzy\subseteq uwv$.$^*$ \hfill $^*$$|x|=|u|$ and $|y|=|v|$

\item[$(I4)$] If $xy\subseteq uv$, then $x\subseteq u$.

\item[$(I5)$] If $xy\subseteq uv$, then $xyy\subseteq uvv$.

\item[$(I6)$] If $x_1x_2\subseteq y_1y_1$ and $z\subseteq vx_2$, then $z\subseteq vx_1$.
\end{enumerate}
\end{definition}

The system extends the one for repetition-free inclusion atoms with the rules $I5$ and $I6$. 

For a set $\Sigma$ of inclusion atoms, we write $x_i\equiv y_j$ to express the \emph{derivable equality} $x_iy_j\subseteq u_lu_l$ whenever $\Sigma\vdash x_iy_j\subseteq u_lu_l$ for some variable $u_l$. We note that in the case of the actual equality $x_1=x_2$, we derive $\vdashx{I1} x_1x_1\subseteq x_1x_1$, hence $x_1\equiv x_1$ and by $x_1=x_2$, $x_1\equiv x_2$. If $\Sigma$ is such that  $x_i\equiv y_j$ is not a derivable equality, we write $x_i\not\equiv y_j$.
In fact, it is not difficult to see that the language of inclusion atoms with repetitions and the language of repetition-free inclusion atoms with (actual) equalities are equally expressive. 


The system in \cite{mitchell} has a more general substitution rule than $I6$, allowing substitution of equal variables on both sides. We show that it is derivable. Further, given a set of equalities between variables, we can reduce any inclusion atom to an equivalent one without repetitions on the right-hand side.

\begin{lemma}\label{substNF}
\begin{enumerate}[align=left, label=(\roman*)]
    \item Substitution on the left-hand side is derivable, i.e.,  $x_1x_2\subseteq y_1y_1, zx_2\subseteq v\vdash zx_1\subseteq v$.\label{substLHS}

        \item Substitution on the right-hand side is derivable from substitution on the left-hand side. 
    \item    Let $\Gamma$ be the set of all derivable equalities from  $u\subseteq v$. Then $u\subseteq v$ is provably equivalent to an atom with distinct variables on the right-hand side, together with $\Gamma$. In particular, $u_1u_2u_3\subseteq y_1y_1v_3\dashv\vdash u_1u_3\subseteq y_1v_3,\, u_1u_2\subseteq y_1y_1$. 
\end{enumerate}
\end{lemma}
\begin{proof}
\begin{enumerate}[align=left,label=(\roman*)]
    \item  We have $\vdashx{I1} zx_1\subseteq zx_1$ and $x_1x_2\subseteq y_1y_1, zx_1\subseteq zx_1\vdashx{I6} zx_1\subseteq zx_2$. We conclude $zx_1\subseteq zx_2,zx_2\subseteq v\vdashx{I2} zx_1\subseteq v$.

    \item Suppose that we have ``If $x_1x_2\subseteq y_1y_1$ and $ zx_2\subseteq v$, then $zx_1\subseteq v$'' as a rule instead of $I6$ in the system. Then $I6$ is derivable, with a similar proof to that of \cref{substLHS}.
    
    \item W.l.o.g., consider the atom $u\subseteq v:= u_1u_2u'\subseteq y_1y_1v'$. 
We have that $u_1u_2u'\subseteq y_1y_1v'$ derives both $u_1u'\subseteq y_1v'$ and $\Gamma$ by $I4$. 
For the other direction, clearly $u_1u_2\subseteq y_1y_1\in \Gamma$, and
%
 $u_1u'\subseteq y_1v'\vdashx{I3} u_1u_1u'\subseteq y_1y_1v'$
 . Putting these together, 
$u_1u_2\subseteq y_1y_1,\, u_1u_1u'\subseteq y_1y_1v' \vdash u_1u_2u'\subseteq y_1y_1v'$  by \cref{substLHS}. 
Hence $u_1u_2u'\subseteq y_1y_1v'\dashv\vdash u_1u'\subseteq y_1v',\Gamma$. 
Iterate this process with the atom $u_1u'\subseteq y_1v'$ and the set $\Gamma$ until there are no more repeated variables in the atom's right-hand side.  


\label{NF}
\end{enumerate}
\qed\end{proof}

We are now ready to provide an alternative completeness proof to the ones in the literature \cite{CASANOVA198429,mitchell}. This alternative proof uses only two values; thus, the completeness result in the Boolean setting follows immediately. 
For any set $\Sigma$ of inclusion atoms, its \emph{generic team} consists of all assignments $s:V\longrightarrow \{0,1\}$ that respect the derivable equalities, and therefore satisfies all atoms in $\Sigma$. In case $\Sigma$ contains infinitely many variables, its generic team is infinite.  We modify the generic teams slightly to obtain the counterexample teams, with \Cref{ct1} and \Cref{ct2} providing two examples. Denote by $0^n$ the $n$-tuple $00\dots 0$.

\begin{theorem}[Completeness]\label{complincl} Let $\Sigma\cup\{x\subseteq y\}$ be a set of inclusion atoms with $|x|=|y|=n$. If $\Sigma\models x\subseteq y$, then $\Sigma\vdash x\subseteq y$.
\end{theorem}
\begin{proof}
If $x=y$, the conclusion follows by $I1$. Assume $x\neq y$ and $\Sigma\not\vdash x\subseteq y$. We define the counterexample team $T$ by $s\in T$ if and only if the following conditions are satisfied:
\begin{enumerate}[label=(\arabic*)]
    \item  $s:V\longrightarrow \{0,1\}$,\label{cond 1}
    \item If $\Sigma\vdash u_i u_j\subseteq u_ku_k$, then $s(u_i)= s(u_j)$.\label{cond 2}
\item[] If there are indecies $j,k$ such that $y_j\equiv y_k$ and $x_j\not\equiv  x_k$, then we stop here. Otherwise, demand also: 
    \item For $w$ such that $\Sigma\vdash w\subseteq y$: $s(w)\neq 0^n$.\label{cond 3}
\end{enumerate}

The decision to choose $0^n$ specifically in condition \ref{cond 3} is motivated by sequences like $0^n$ never being at risk of removal by condition \ref{cond 2}, ensuring that $0^n\in T[x]$.
Now condition \ref{cond 3} implies that $T \not\models x\subseteq y$. 

Let $u\subseteq v\in \Sigma$, we show that $T \models u\subseteq v$.
First, consider the cases where the team is constructed without condition \ref{cond 3}, or when $V(w)\not\subseteq V(v)$ for all $w$ such that $\Sigma\vdash w\subseteq y$. 
If $v$ does not contain variables affected by condition \ref{cond 2}, then all $01$-combinations are in $T[v]$ and $T\models u\subseteq v$ follows. 
If there are variables in $v$ such that for all $s\in T$, $s(v_i)= s(v_j)$, by $I2$ also $s(u_i)= s(u_j)$ for all $s\in T$. Now any value missing from $T[v]$ is also missing from $T[u]$, hence $T\models u\subseteq v$.

Else, for each $w$, $V(w)\subseteq V(v)$, such that $\Sigma\vdash w\subseteq y$, we have $0^n\not\in T[w]$ by condition \ref{cond 3}. We derive $u\subseteq v \vdash u'\subseteq w$ by $I3$, $I4$ and $I5$, where $V(u')\subseteq V(u)$. By $I2$, $\Sigma\vdash u'\subseteq y$. Since $\Sigma\not\vdash x\subseteq y$, we have that $u'\neq x$ also after substituting derivable equalities using $I6$ and \Cref{substNF} \cref{substLHS}. Therefore, by condition \ref{cond 3}, $0^n\not\in T[u']$. The effect of condition \ref{cond 2} is like in the previous case, hence we conclude $T\models u\subseteq v$. 
\qed\end{proof}

For repetition-free inclusion atoms, the completeness proof of \Cref{complincl} without condition \ref{cond 2} suffices for the repetition-free system $I1$-$I4$.

\begin{table}[ht]
\caption{Illustrated in the four tables are all possible assignments to the values $0$ and $1$ over the variables in $\{x_1,x_2,y_1,y_2,z_1\}$. The counterexample team, according to the proof of \Cref{complincl}, for the consequence $x_1x_2\subseteq y_1 y_2$ and assumption set $\Sigma=\{y_1y_2\subseteq z_1z_1,\,  x_1z_1\subseteq z_1z_1\}$ consists of all non-crossed out lines. Since $y_1\equiv y_2$ while $x_1\not\equiv x_2$, condition \ref{cond 3} is not used in the construction of the team.} \label{ct1}
\parbox{.24\textwidth}{
\centering
$\begin{array}{ccccc}
\toprule	
x_1 & x_2 &y_1& y_2 &z_1 \\ 
 \midrule
1 & 1 &1& 1 &1 \\ 
 1& 0 &1& 1 &1 \\ 
\rowcolor{black!15} \tikzmark{start17} 0& 1 &1& 1 &1 \tikzmark{start18}\\ 
\rowcolor{black!15}0 & 0 &1& 1 &1\\ 
\rowcolor{black!15} 1 & 1 &1& 1 &0 \\ 
\rowcolor{black!15} \tikzmark{end18}1 & 0 &1& 1 &0 \tikzmark{end17}\\ 
0 & 1 &1& 1 &0 \\ 
0 & 0 &1& 1 &0 \\ 
\bottomrule 
\end{array}$ 
  \tikz[remember picture] \draw[overlay] ([yshift=.35em]pic cs:start17) -- ([yshift=.35em]pic cs:end17);
   \tikz[remember picture] \draw[overlay] ([yshift=.35em]pic cs:start18) -- ([yshift=.35em]pic cs:end18);
}
\parbox{.25\textwidth}{
\centering
$\begin{array}{ccccc}
\toprule	
x_1 & x_2 &y_1& y_2 &z_1 \\ 
 \midrule
\rowcolor{black!15}\tikzmark{start15}1 & 1 &1& 0 &1 \tikzmark{start16}\\ 
\rowcolor{black!15} 1& 0 &1& 0 &1 \\ 
\rowcolor{black!15} 0& 1 &1& 0 &1 \\ 
\rowcolor{black!15}0 & 0 &1& 0 &1\\ 

\rowcolor{black!15}1 & 1 &1& 0 &0 \\ 
\rowcolor{black!15} 1 & 0 &1& 0 &0 \\ 
\rowcolor{black!15}0 & 1 &1&0 &0 \\ 
\rowcolor{black!15}\tikzmark{end16}0 & 0 &1& 0 &0 \tikzmark{end15}\\ 
\bottomrule 
\end{array}$ 
  \tikz[remember picture] \draw[overlay] ([yshift=.35em]pic cs:start15) -- ([yshift=.35em]pic cs:end15);
   \tikz[remember picture] \draw[overlay] ([yshift=.35em]pic cs:start16) -- ([yshift=.35em]pic cs:end16);
}
\parbox{.25\textwidth}{
\centering
$\begin{array}{ccccc}
\toprule	
x_1 & x_2 &y_1& y_2 &z_1 \\ 
 \midrule
\rowcolor{black!15}\tikzmark{start13}1 & 1 &0 &1 &1 \tikzmark{start14}\\ 
\rowcolor{black!15}1 & 0 &0 &1 &1 \\ 
\rowcolor{black!15}0 & 1 &0 &1 &1 \\ 
\rowcolor{black!15}0 & 0 &0 &1 &1 \\ 
\rowcolor{black!15}1 & 1 &0 &1 &0 \\ 
\rowcolor{black!15}1 & 0 &0 &1 &0 \\ 
\rowcolor{black!15}0 & 1 &0 &1 &0 \\ 
\rowcolor{black!15}\tikzmark{end14}0 & 0 &0 &1 &0 \tikzmark{end13}\\ 
\bottomrule 
\end{array}$ 
  \tikz[remember picture] \draw[overlay] ([yshift=.35em]pic cs:start13) -- ([yshift=.35em]pic cs:end13);
   \tikz[remember picture] \draw[overlay] ([yshift=.35em]pic cs:start14) -- ([yshift=.35em]pic cs:end14);
}
\parbox{.24\textwidth}{
\centering
$\begin{array}{ccccc}
\toprule	
x_1 & x_2 &y_1& y_2 &z_1 \\ 
 \midrule
1 & 1 &0& 0 &1 \\ 
1 & 0 &0& 0 &1 \\ 
\rowcolor{black!15}\tikzmark{start11}0 & 1 &0& 0 &1 \tikzmark{start12}\\ 
\rowcolor{black!15}0 & 0 &0& 0 &1 \\ 
\rowcolor{black!15}1 & 1 &0& 0 &0 \\ 
\rowcolor{black!15}\tikzmark{end12}1 & 0 &0& 0 &0 \tikzmark{end11}\\ 
0 & 1 &0& 0 &0 \\ 
0 & 0 &0& 0 &0 \\ 
\bottomrule 
\end{array}$ 
  \tikz[remember picture] \draw[overlay] ([yshift=.35em]pic cs:start11) -- ([yshift=.35em]pic cs:end11);
   \tikz[remember picture] \draw[overlay] ([yshift=.35em]pic cs:start12) -- ([yshift=.35em]pic cs:end12);
}
\end{table}
%
%
\begin{table}[ht]
\caption{The counterexample team as per the proof of \Cref{complincl} for the consequence $x_1x_2\subseteq y_1y_2$ and assumption set $\Sigma=\{y_2z_1\subseteq z_1z_1, \, x_1z_1\subseteq y_1y_2\}$ consists of all non-crossed out lines from the four subteams illustrated below. 
}  \label{ct2}
\parbox{.24\textwidth}{
\centering
$\begin{array}{ccccc}
\toprule	
x_1 & x_2 &y_1& y_2 &z_1 \\ 
 \midrule
1 & 1 &1& 1 &1 \\ 
 1& 0 &1& 1 &1 \\ 
 0& 1 &1& 1 &1 \\ 
0 & 0 &1& 1 &1\\ 

\rowcolor{black!15}\tikzmark{start9}1 & 1 &1& 1 &0 \tikzmark{start10}\\ 
\rowcolor{black!15} 1 & 0 &1& 1 &0 \\ 
\rowcolor{black!15}0 & 1 &1& 1 &0 \\ 
\rowcolor{black!15}\tikzmark{end10}0 & 0 &1& 1 &0 \tikzmark{end9}\\ 
\bottomrule 
\end{array}$ 
  \tikz[remember picture] \draw[overlay] ([yshift=.35em]pic cs:start9) -- ([yshift=.35em]pic cs:end9);
   \tikz[remember picture] \draw[overlay] ([yshift=.35em]pic cs:start10) -- ([yshift=.35em]pic cs:end10);
}
\parbox{.25\textwidth}{
\centering
$\begin{array}{ccccc}
\toprule	
x_1 & x_2 &y_1& y_2 &z_1 \\ 
 \midrule
\rowcolor{black!15} \tikzmark{start5}1 & 1 &1& 0 &1 \tikzmark{start6}\\ 
\rowcolor{black!15} 1& 0 &1& 0 &1 \\ 
\rowcolor{black!15} 0& 1 &1& 0 &1 \\ 
\rowcolor{black!15}\tikzmark{end6}0 & 0 &1& 0 &1\tikzmark{end5}\\ 

1 & 1 &1& 0 &0 \\ 
1 & 0 &1& 0 &0 \\ 
\rowcolor{black!15}\tikzmark{start8}0 & 1 &1&0 &0 \tikzmark{start7}\\ 
\rowcolor{black!15}\tikzmark{end7}0 & 0 &1& 0 &0 \tikzmark{end8} \\ 
\bottomrule 
\end{array}$ 
 \tikz[remember picture] \draw[overlay] ([yshift=.35em]pic cs:start5) -- ([yshift=.35em]pic cs:end5);
   \tikz[remember picture] \draw[overlay] ([yshift=.35em]pic cs:start6) -- ([yshift=.35em]pic cs:end6);
  \tikz[remember picture] \draw[overlay] ([yshift=.35em]pic cs:start7) -- ([yshift=.35em]pic cs:end7);
   \tikz[remember picture] \draw[overlay] ([yshift=.35em]pic cs:start8) -- ([yshift=.35em]pic cs:end8);

}
\parbox{.25\textwidth}{
\centering
$\begin{array}{ccccc}
\toprule	
x_1 & x_2 &y_1& y_2 &z_1 \\ 
 \midrule
1 & 1 &0 &1 &1 \\ 
1 & 0 &0 &1 &1 \\ 
0 & 1 &0 &1 &1 \\ 
0 & 0 &0 &1 &1 \\ 
\rowcolor{black!15}\tikzmark{start3}1 & 1 &0 &1 &0 \tikzmark{start4}\\ 
\rowcolor{black!15}1 & 0 &0 &1 &0 \\ 
\rowcolor{black!15}0 & 1 &0 &1 &0 \\ 
\rowcolor{black!15}\tikzmark{end4}0 & 0 &0 &1 &0 \tikzmark{end3}\\ 
\bottomrule 
\end{array}$ 
 \tikz[remember picture] \draw[overlay] ([yshift=.35em]pic cs:start3) -- ([yshift=.35em]pic cs:end3);
  \tikz[remember picture] \draw[overlay] ([yshift=.35em]pic cs:start4) -- ([yshift=.35em]pic cs:end4);
}
\parbox{.24\textwidth}{
\centering
$\begin{array}{ccccc}
\toprule	
x_1 & x_2 &y_1& y_2 &z_1 \\ 
 \midrule
\rowcolor{black!15} \tikzmark{start1}1 & 1 &0& 0 &1 \tikzmark{start2}\\ 
\rowcolor{black!15}1 & 0 &0& 0 &1 \\ 
\rowcolor{black!15}0 & 1 &0& 0 &1 \\ 
\rowcolor{black!15}0 & 0 &0& 0 &1 \\ 
\rowcolor{black!15}1 & 1 &0& 0 &0 \\ 
\rowcolor{black!15}1 & 0 &0& 0 &0 \\ 
\rowcolor{black!15}0 & 1 &0& 0 &0 \\ 
\rowcolor{black!15}\tikzmark{end2}0 & 0 &0& 0 &0 \tikzmark{end1}  \\ 
\bottomrule 
\end{array}$ 
 \tikz[remember picture] \draw[overlay] ([yshift=.35em]pic cs:start1) -- ([yshift=.35em]pic cs:end1);
   \tikz[remember picture] \draw[overlay] ([yshift=.35em]pic cs:start2) -- ([yshift=.35em]pic cs:end2);
 }
\end{table}

Similarly to the decision problem for repetition-free inclusion atoms \cite{CASANOVA198429}, the decision problem for inclusion atoms is PSPACE-complete. 

\begin{theorem}[Complexity]
   Let $\Sigma\cup \{x\subseteq y\}$ be a finite set of inclusion atoms. Deciding whether $\Sigma\models x\subseteq y$ is PSPACE-complete. \label{PSPACE_rep}
\end{theorem}

\begin{proof}
  PSPACE-hardness is immediate by the corresponding result for repetition-free inclusion in \cite{CASANOVA198429}. We sketch the proof for PSPACE membership.

 First, identify equal variables by scanning all of $\Sigma$ for repeated variables on the right-hand side of the atoms. Collect these into a list of pairs, e.g., $(x_2,x_2)$, $(z_1,z_1)$.
 Scan $\Sigma$ again, and if the variables from some pair on the list both appear on the right-hand side of an atom, add the pair of corresponding variables on the left-hand side to the list, e.g. if $z_2yu_3\subseteq x_2yx_2\in\Sigma$ then we add $(z_2,u_3)$ (motivated by the rules $I2$-$I4$).
 Iterate this process until no more pairs can be added to the list. Since $\Sigma$ is finite,
 the procedure necessarily terminates in at most $|Var|^2$ steps, where each step is in linear time. 

    With the list of equalities, construct $\Sigma^*$ by replacing all equal variables in $\Sigma$ with exactly one of them, and modify the atoms further such that no variable is repeated on the right-hand side of the atoms, as outlined in \Cref{substNF} \cref{NF}. Rewrite also $x\subseteq y$ in the same way, and denote it by $x\subseteq y^*$. 
    
    Now, with the equalities given, $\Sigma^*$ is equivalent to $\Sigma$,    
    and contains no repetitions of variables on the right-hand side of its atoms. Since the complete system for inclusion atoms has no rule specific to repetitions on the left-hand side, we can use the usual PSPACE algorithm in \cite{CASANOVA198429} to check whether $\Sigma^*$ entails $x\subseteq y^*$.
\qed\end{proof}

We end this section by showing that, unlike inclusion atoms in the first-order setting \cite{inclARMST}, not even unary inclusion atoms have Armstrong relations in the Boolean setting. 
We say that a set of dependencies $\Sigma$ has Armstrong relations if for all $\Sigma'\subseteq\Sigma$ there is a team such that only the atoms entailed by $\Sigma'$ are satisfied.

\begin{theorem}[Nonexistence of Armstrong relations] There is a set $\Sigma'$ of propositional inclusion atoms such
that no team satisfies exactly the propositional inclusion atoms that $\Sigma'$ semantically entails.
\end{theorem}

\begin{proof}
    Let $Prop=\{p_1,p_2,p_3\}$ and $\Sigma'=\{p_1\subseteq p_2\}$. We show that any team satisfying $\Sigma'$ must satisfy either $p_3\subseteq p_2$ or $p_2\subseteq p_1$, neither of which is entailed by $\Sigma'$.     If $T\models p_3\subseteq p_2$, we are done. So suppose that $T\models\Sigma'$ and $T\not\models p_3\subseteq p_2$. Then there is a value $a\in \{0,1\}$ such that $a\in T[p_3]\setminus T[p_2]$. W.l.o.g. let $a=0$. Now by $T\models p_1\subseteq p_2$, we have that $T[p_1]\subseteq T[p_2]=\{1\}$. Thus, $T\models p_2\subseteq p_1$.
\qed\end{proof}

\section{Inclusion with Boolean constants} \label{sec bool}

We remain in the Boolean setting and consider inclusion atoms with Boolean constants of the form $p_1\dots p_n\subseteq q_1\dots q_n$, where $p_1\dots p_n,q_1\dots q_n$ are either propositional variables or the constants $\top$ and $\bot$. We introduce a complete system, show that the decision problem is PSPACE-complete, and conclude that there is no complete $k$-ary system for inclusion atoms with Boolean constants.

We extend the scope of assignments to the constants $\top$ and $\bot$.

\begin{definition}
For any assignment $s:Var\cup\{\top,\bot\}\longrightarrow \{0,1\}$, we have that 
$s(\top)=1$ and $s(\bot)=0$ always hold. The semantic clause for inclusion atoms with Boolean constants is as before: 
\begin{equation*}T\models p\subseteq q \text{ if and only if for all }
s\in T,\text{ there exists } s'\in T \text{ such that } s(p)=s'(q).\end{equation*}
\end{definition}

In this setting, in addition to trivial atoms, we also have \emph{contradictory} atoms that are only satisfied in the empty team, e.g.,  $\bot\subseteq \top$. We cannot express contradictions with the usual propositional inclusion atoms, making the inclusion atoms with Boolean constants strictly more expressive.

We define the system for inclusion atoms with Boolean constants by extending the system in \Cref{Inclusion full rules}.
Hereforth we exclusively use $\mathsf{x}:=\mathsf{x}_1\dots\mathsf{x}_n$ to denote a sequence of constants $\top$ and $\bot$. Otherwise, we do not exclude the constants $\top$ and $\bot$ from the sequences $p,q,r$, or any other variable sequence, unless explicitly mentioned.

We say that $\mathsf{x}_1\dots\mathsf{x}_n\subseteq r_1\dots r_n$ is \emph{consistent} if for all $i,j\in \{1,\dots, n\}$, $r_i\in \{\top,\bot\}$ implies $r_i= \mathsf{x}_i$, and additionally, if $r_i=r_j$, then $\mathsf{x}_i=\mathsf{x}_j$. Now $\top\bot\subseteq \top p_2$ and $\top\top\subseteq p_2 p_2$ are consistent, while  $\top\bot\subseteq\bot p_2$ and $\top\bot\subseteq p_2 p_2$ are not.

\begin{definition}\label{rulesBoolean}
The rules for inclusion atoms with Boolean constants are the rules $I1$-$I6$ from \Cref{Inclusion full rules} together with $B1$, $B2$ and $B3$.

    \begin{enumerate}[align=left]
\item[$(B1)$] If $\top\subseteq \bot$, then $q\subseteq r$. 

\item[$(B2)$] If $p\subseteq q$, then $p\top\subseteq q \top$ and $p\bot\subseteq q \bot$.

\item[$(B3)$] Let $n:=|p|$ and $A\subseteq\{r\subseteq q\mid r_i\in \{p_i, \top,\bot\}, 1\leq i\leq n\}$ be a minimal set such that for any $\mathsf{x}$ with consistent $\mathsf{x}\subseteq p$, there is some $r\subseteq q \in A$ such that $\mathsf{x}\subseteq r$ is consistent. If $A$, then $p \subseteq q$. 
\end{enumerate} 
\end{definition}

The rule $B1$ allows us to derive any atom from $\top\subseteq\bot$, since it is only satisfied by the empty team. We also derive any atom from the assumption $\bot\subseteq\top$, since we have $\vdashx{I1} \top\subseteq\top$, and $\bot\subseteq\top\vdashx{B2} \bot\top\subseteq\top\top$, thus $\top\subseteq\bot$ follows by $I6$. The rule $B2$ is from the larger system for propositional inclusion logic in \cite{yang_propositional_2022}, and allows us to add constants to both sides of any inclusion atom. Note that in general $p\subseteq q\not\models pr\subseteq qr$. 

The schema $B3$ has many instances, and we will closely examine one of the more involved ones in \Cref{nonfiniteaxiom}. A simple application of the schema is if all $\top\bot$-combinations are included in $q$, then any sequence is included in $q$. Another example is that $p\subseteq q$ follows from $\top p_2\dots p_n\subseteq q_1\dots q_n$ and $\bot p_2\dots p_n\subseteq q_1\dots q_n$.

\begin{lemma}
    The rules $B1$, $B2$ and axiom schema $B3$ are sound.
\end{lemma}

\begin{proof}
Proving soundness of $B1$ and $B2$ is straightforward. We show the claim for $B3$. Let $A$ be as stated in the premise of the rule, and assume that $T\models r \subseteq q$, for all $r\subseteq q\in A$. Let $s\in T$, and consider the sequence $\mathsf{x}$ for which $s(p)=s(\mathsf{x})$. By the definition of $A$, there is some $r\subseteq q\in A$ such that $\mathsf{x}\subseteq r$ is consistent. Thus we find some $s'\in T$, for which  $s(p)= s(\mathsf{x})=s'(r)\subseteq T[q]$. Hence, $T\models p \subseteq q$.
\qed\end{proof}

We generalise the completeness proof of \Cref{complincl} to inclusion atoms with Boolean constants. \Cref{ct3} shows an instance of the counterexample team.

\begin{theorem}[Completeness]\label{CompletenessBoolean}
 Let $\Sigma\cup \{p\subseteq q\}$ be a set of inclusion atoms with Boolean constants. If $\Sigma\models p\subseteq q$, then $\Sigma\vdash p\subseteq q$.
\end{theorem}

\begin{proof}
Suppose that $\Sigma\not\vdash p\subseteq q$. Then both $p=q$ and $\Sigma\vdash \top\subseteq\bot$ would be contradictory by $I1$ and $B1$, so suppose neither is the case. Let $n:=|p|$. 

We first show that under the assumption $\Sigma\not\vdash p\subseteq q$, there is a sequence $\mathsf{x}$ such that $\Sigma\not\vdash \mathsf{x}\subseteq q$ and $\mathsf{x}\subseteq p$ is consistent. Consider sequences $r$ such that $r_i\in \{p_i, \top,\bot\}$, $1\leq i\leq n$. 
If there is no $r$ such that $\Sigma\vdash r\subseteq q$, then let $c:=s(\mathsf{x})$ for some $\mathsf{x}$ for which $\mathsf{x}\subseteq p$ is consistent, where $s$ is any assignment. Otherwise, lest we derive $\Sigma\vdash p\subseteq q$ by $B3$, there is some $n$-sequence $\mathsf{x}$ such that $\mathsf{x}\subseteq p$ is consistent but $\mathsf{x}\subseteq r$ is not consistent for any $r$ with $\Sigma\vdash r\subseteq q$. In particular, $\Sigma\not\vdash \mathsf{x}\subseteq q$, so we set $c:=s(\mathsf{x})$, where again $s$ is any assignment.

Build the team $T$ by letting $s\in T$ if the following conditions are met.
\begin{enumerate}[label=(\arabic*)]
    \item $s:V\cup\{\top,\bot\}\longrightarrow \{0,1\}$,\label{Cond 1}
    \item If $\Sigma\vdash u_i u_j\subseteq u_ku_k$, then $s(u_i)= s(u_j)$.\label{Cond 2}
    \item[] If there are indecies $j,k$ such that $q_j\equiv q_k$ and $p_j\not\equiv p_k$, then we stop here. Otherwise, demand also: 
    \item For $w$ such that $\Sigma\vdash w\subseteq q$: $s(w)\neq c$. \label{Cond 3}
\end{enumerate} 

Condition \ref{Cond 2} handles equalities not only between variables, but also between variables and constants, since, e.g., $\Sigma\vdash u_i\subseteq\top \vdashx{B2} u_i\top\subseteq\top\top$.

The intuition as to why the counterexample team works is the same as in the case of \Cref{complincl}: it is as close to the generic team as possible, i.e., it satisfies conditions \ref{Cond 1} and \ref{Cond 2}, with a minimal intervention by condition \ref{Cond 3} ensuring that $T\not\models p\subseteq q$. Thus, showing that $T \models u\subseteq v$ for all $u\subseteq v\in \Sigma$ is analogous to the proof of \Cref{complincl}.
\qed\end{proof}

\begin{table}[ht]
\caption{The counterexample team in the proof of \Cref{CompletenessBoolean} for the consequence $p_1p_2\subseteq q_1 q_2$ and assumption set $\Sigma=\{\bot\bot\subseteq q_1q_2,\, \top\top\subseteq q_1q_2,\, \bot\top\subseteq q_1q_2\}$ consists of all the non-crossed out lines. Here, condition \ref{Cond 3} is applied with $c= 10$.}\label{ct3}
\parbox{.24\textwidth}{
\centering
$\begin{array}{ccccc}
\toprule	
p_1& p_2 &q_1& q_2 &r_1 \\ 
 \midrule
1 & 1 &1& 1 &1 \\ 
 1& 0 &1& 1 &1 \\ 
 0& 1 &1& 1 &1 \\ 
0 & 0 &1& 1 &1\\ 

1 & 1 &1& 1 &0 \\ 
1 & 0 &1& 1 &0 \\ 
0 & 1 &1& 1 &0 \\ 
0 & 0 &1& 1 &0 \\ 
\bottomrule 
\end{array}$ 
}
\parbox{.25\textwidth}{
\centering
$\begin{array}{ccccc}
\toprule	
p_1 & p_2 &q_1& q_2 &r_1 \\ 
 \midrule
\rowcolor{black!15}\tikzmark{start20}1 & 1 &1& 0 &1 \tikzmark{start21}\\ 
\rowcolor{black!15} 1& 0 &1& 0 &1 \\ 
\rowcolor{black!15} 0& 1 &1& 0 &1 \\ 
\rowcolor{black!15}0 & 0 &1& 0 &1\\ 

\rowcolor{black!15}1 & 1 &1& 0 &0 \\ 
\rowcolor{black!15}1 & 0 &1& 0 &0 \\ 
\rowcolor{black!15}0 & 1 &1&0 &0 \\ 
\rowcolor{black!15}\tikzmark{end21}0 & 0 &1& 0 &0 \tikzmark{end20}\\ 
\bottomrule 
\end{array}$ 
 \tikz[remember picture] \draw[overlay] ([yshift=.35em]pic cs:start20) -- ([yshift=.35em]pic cs:end20);
   \tikz[remember picture] \draw[overlay] ([yshift=.35em]pic cs:start21) -- ([yshift=.35em]pic cs:end21);
}
\parbox{.25\textwidth}{
\centering
$\begin{array}{ccccc}
\toprule	
p_1 & p_2 &q_1& q_2 &r_1 \\ 
 \midrule
1 & 1 &0 &1 &1 \\ 
1 & 0 &0 &1 &1 \\ 
0 & 1 &0 &1 &1 \\ 
0 & 0 &0 &1 &1 \\ 
1 & 1 &0 &1 &0 \\ 
1 & 0 &0 &1 &0 \\ 
0 & 1 &0 &1 &0 \\ 
0 & 0 &0 &1 &0 \\ 
\bottomrule 
\end{array}$ 
}
\parbox{.24\textwidth}{
\centering
$\begin{array}{ccccc}
\toprule	
p_1 & p_2 &q_1& q_2 &r_1 \\ 
 \midrule
1 & 1 &0& 0 &1 \\ 
1 & 0 &0& 0 &1 \\ 
0 & 1 &0& 0 &1 \\ 
0 & 0 &0& 0 &1 \\ 
1 & 1 &0& 0 &0 \\ 
1 & 0 &0& 0 &0 \\ 
0 & 1 &0& 0 &0 \\ 
0 & 0 &0& 0 &0 \\ 
\bottomrule 
\end{array}$ 
}\end{table}

Next, we examine the axiom schema $B3$ from a computational complexity standpoint.

\begin{lemma}\label{co-NP-compl}
    Deciding whether $A\vdashx{B3} p\subseteq q$ is co-NP complete.
\end{lemma}

\begin{proof}
 Let $A$ be an assumption set as stated in the schema $B3$, and let $|p|=n$. 

For co-NP-hardness, we reduce the co-NP-complete (\cite{cook}) negation of the 3-SAT problem. Let $\phi=(a_1\lor a_2 \lor a_3)\land \dots \land  (a_{n-2}\lor a_{n-1}\lor a_{n})$ be a conjunction of Boolean disjunctive clauses with three disjuncts, where $a_i\in \{p_i, \neg p_i\}$, $1\leq i\leq n$. Note that we allow $p=p_1\dots p_n$ to have repeated variables.   

To determine whether $\phi\not\in SAT$, associate assignments $s$ over the variables in $p$ to sequences $\mathsf{x}$ in the obvious way: $\mathsf{x}_i=\top$ iff $s(p_i)=1$.  For $1\leq i\leq n$, define 
 \[   
\overbar{a_i} = 
     \begin{cases}
       \bot &\quad\text{if } a_i=p_i,\\
       \top &\quad\text{if } a_i=\neg p_i. \\
     \end{cases}
\]

Let $r$ be a $n$-sequence of fresh propositional symbols. The reduction is completed through the equivalences between the following statements.
\begin{enumerate} [label=\textbf{-}]

    \item  $(a_1\lor a_2 \lor a_3)\land \dots \land  (a_{n-2}\lor a_{n-1}\lor a_{n})\not\in SAT$.

 \item For all assignments $s$ corresponding to the sequence $\mathsf{x}$, there is some $i\in \{1,4,\dots, n-2\}$ such that $\mathsf{x}_i \mathsf{x}_{i+1} \mathsf{x}_{i+2}= \overbar{a_i}\overbar{ a_{i+1}} \overbar{a_{i+2}}$.

 \item  For all $\mathsf{x}$ where $\mathsf{x}\subseteq p$ is consistent, there is some $i\in \{1,4,\dots, n-2\}$ such that $\mathsf{x}_i \mathsf{x}_{i+1} \mathsf{x}_{i+2} = \overbar{a_i}\overbar{ a_{i+1}} \overbar{a_{i+2}}$.

    \item  $A\vdashx{B3} p\subseteq r$, where     
    \[
  A = \left\{ b\subseteq r\ \middle\vert \begin{array}{l}
    \text{there is exactly one $i\in \{1,4,\dots, n-2\}$ such that } \\
\text{$b_ib_{i+1}b_{i+2}=\overbar{a_i}\overbar{a_{i+1}} \overbar{a_{i+2}}$
    and $b_j=p_j$ when $j\not\in\{i,i+1,i+2\}$} 
  \end{array}\right\}.
\]
\end{enumerate}

For co-NP-membership, it suffices to show that $A\not\vdashx{B3} p\subseteq q$ is in NP. We sketch the proof. Guess a sequence $\mathsf{x}$ for which $x\subseteq p$ is consistent. Check whether there is some $u$ appearing as the left-hand side of some atom in $A$ such that $\mathsf{x}\subseteq u$ is consistent. If no, the algorithm accepts $\mathsf{x}$ as the witness, and we can find a certificate for $A\not\vdashx{B3} p\subseteq q$ in polynomial time. 
\qed\end{proof}

We conclude this section by showing that the decision problem for inclusion atoms with Boolean constants is PSPACE-complete. 

\begin{theorem}[Complexity]
   Let $\Sigma\cup \{p\subseteq q\}$ be a finite set of inclusion atoms with Boolean constants. Deciding whether $\Sigma\vdash p\subseteq q$ is PSPACE-complete.
\end{theorem}

\begin{proof}

PSPACE-hardness follows from the corresponding result for inclusion atoms in \cite{CASANOVA198429}. We give a sketch of PSPACE-membership. Extend the procedure in the proof of \Cref{PSPACE_rep} to also list equalities between variables and constants, by adding, e.g., $(p_1,\top)$ to the list if there is an atom with $p_1$ on the left-hand side in the same relative position as $\top$ on the right-hand side. We again modify $\Sigma$ to $\Sigma^*$, such that no atom in $\Sigma^*$ has repeated variables or constants on its right-hand side. Given the list of equalities, the sets $\Sigma$ and $\Sigma^*$ are equivalent by \Cref{substNF} \cref{NF} and $B2$. We also modify $p\subseteq q$ in the same way and denote it by $p\subseteq q^*$.

Now it suffices to check whether $\Sigma^*\models  p\subseteq q^*$, hence we only need to extend the PSPACE algorithm in \cite{CASANOVA198429} to entailments with constants on the left-hand side in the atoms, i.e., entailments corresponding to $B3$, which by \Cref{co-NP-compl} do not exceed PSPACE. 
\qed\end{proof}

\subsection{No $k$-ary axiomatization}

The complete proof system defined in \Cref{rulesBoolean} is large, due to the axiom schema $B3$ not having an upper bound on the number of assumptions its instances can require. We show that this is necessary, i.e., no complete $k$-ary proof system exists for inclusion atoms with Boolean constants.

Let $p$ and $q$ be $n$-sequences of distinct propositional variables. Define $B$ by 
  \[
  B = \left\{ r\subseteq q\ \middle\vert \begin{array}{l}
      \text{$r=\bot^n$, or there is exactly one $i\in \{1,\dots, n\}$} \\ \text{such that $r_i=\top$
    and  $r_j=p_j$ when $j\neq i$\, }  
  \end{array}\right\}.
\]

We consider the following instances of the schema $B3$: \begin{equation*} 
   \text{If $B$, then $p\subseteq q$.}\label{rule}\tag{*}
\end{equation*}

The number of assumptions, i.e., the arity, of the rule (\ref{rule}) is $n+1$ for atoms of arity $n$, and we show that it cannot be reduced. It follows that the arity of the system grows with the arity of the atoms considered. 

We show the claim by following the strategy in \cite{cond_ind}, and prove that no assumption can be derived from the others and that no non-trivial atom can be concluded from the assumption set, except for the conclusion of the rule. This is a general strategy to show the absence of $k$-ary axiomatizations, as per \cite{CASANOVA198429}. We end the section by sketching the proof.   

\begin{theorem}\label{nonfiniteaxiom}
    There is no $k$-ary axiomatization of inclusion atoms with Boolean constants.
\end{theorem}
\begin{proof}
Consider the $n$-ary atom $p\subseteq q$, and let $B$ be as in the statement of the rule (\ref{rule}). Then $B\vdash p\subseteq q$. We show that any complete proof system necessarily includes a rule for this derivation with $n+1$ assumptions.
It suffices to prove the following two claims: 
\begin{enumerate}
    \item[(1)] No assumption is derived from the others.
    \item[(2)] The only nontrivial atom derivable from the assumption set $B$ is the conclusion  $p\subseteq q$.
\end{enumerate}

We easily build teams witnessing (1).
For (2), we reduce the set of nontrivial atoms we need to consider to ones of the form $u\subseteq q\neq p\subseteq q$. It suffices to show that for each $u\subseteq q$, there is a team that satisfies the atoms in $B$ but not $u\subseteq q$. There are many distinct cases based on the configuration of the variables and constants appearing in $u$. Still, two types of teams can be used (with some modification) in all cases. The first team is $T_1=\{s_1, s_2\}$ with $s_1(pq)=1^{2n}$ and $s_2(pq)=1^n0^n$. The second team $T_2=\{s_1,\dots s_{n+1}\}$ is defined by $T_2=\{s \mid s(p)=0^n \text{ and there is at most one $1\leq i\leq n$ such that } s(q_i)=1\}$. The team $T_1$ covers the cases when there are $i,j$ such that $u_i=\bot$ and $u_j\neq\bot$, and $T_2$ the cases when there are different $i,j$ such that $u_i=u_j=\top$. Both teams are illustrated in \Cref{team A}.\qed\end{proof}

\begin{table}[ht]
\centering
\caption{The teams $T_1$ and $T_2$ from the proof of \Cref{nonfiniteaxiom}.}\label{team A}
$\begin{array}{c cccc cccc}
\toprule	
T_1:\,\,&p_1 & p_2 & \dots & p_n &q_1& q_2 & \dots &q_n \\ 
 \midrule
&1&1 & \dots &1& 1 &1& \dots &1 \\ 
&1&1 & \dots &1& 0 &0& \dots &0 \\ 
\bottomrule 
\end{array}$ \hspace{.5cm}
$\begin{array}{c cccc cccc}
\toprule	
T_1:\,\,& p_1 & p_2 & \dots & p_n &q_1& q_2 & \dots &q_n \\ 
 \midrule
&0&0 & \dots &0&     1 &0& \dots&0 \\ 
&0&0 &\dots &0&     0 &1& \dots &0 \\ 
&& &\ddots &&      && \ddots & \\ 
&0&0 & \dots &0&     0 &0& \dots& 1 \\ 
&0&0 &\dots &0&     0 &0& \dots &0 \\ 
\bottomrule 
\end{array}$ 
\end{table}

\section{Conclusion and future work} \label{sec conc}

We list the paper's key contributions regarding the implication problem for inclusion with repetitions and Boolean constants, and subsequently suggest directions for future work.  
\begin{enumerate}[label=$\bullet$]
    \item Providing an alternative completeness proof of standard inclusion dependencies for the usual system \cite{CASANOVA198429}, using only two values, thereby demonstrating completeness also in the Boolean setting.

    \item 
Extending the alternative completeness proof to inclusion atoms with repetitions for a system extracted from \cite{mitchell}.

\item Introducing a complete proof system for inclusion atoms with Boolean constants and confirming that no such system is $k$-ary.


\item  Demonstrating that the decision problems for both extended inclusion atoms remain PSPACE-complete.
\end{enumerate}

The completeness proofs for the systems of functional dependence \cite{Galliani2014}, non-conditional independence \cite{Galliani2014}, and now also inclusion, only use two values, making the systems complete when restricted to the Boolean setting. This is not always the case; there are semantic entailments valid only in the Boolean setting for anonymity and exclusion atoms (see \cite{Vaananen2022} and \cite{haggblom2024} for their definitions). For anonymity and exclusion, the following entailments serve as examples that are only sound in the Boolean setting: $q\Upsilon p, p\Upsilon r_1\models r_1\Upsilon p$ and $p_1p_2|p_2p_1\models p_1|p_2$. Extending the systems in \cite{Vaananen2022,haggblom2024} to complete systems in the Boolean setting for anonymity and exclusion atoms remains as future work.

Another use of our alternative completeness proof is in the axiomatization of quantity approximate inclusion atoms of the form $x\subseteq_n y$ in \cite{haggblom2025}, allowing $n$ values of $x$ to be missing from $y$. The alternative counterexample teams are advantageous compared to the existing ones when the degree to which they should not satisfy the assumed non-derivable atom needs to be carefully controlled.

Inclusion atoms with repetitions can express equalities between variables, making this a natural dependency class to consider. Similarly, for exclusion atoms, allowing repetitions increases the expressivity, and the rules for repetition-free exclusion in \cite{Casanova1983TowardsAS} must be extended to obtain the complete system provided in \cite{haggblom2024}. We could thus consider extending the system for the combined implication problem for repetition-free inclusion and exclusion in \cite{Casanova1983TowardsAS} to deal with repetitions. Additionally, the language of team-based propositional logic augmented with inclusion atoms is axiomatized in \cite{yang_propositional_2022}, creating a direction of future study: the axiomatization of propositional inclusion-exclusion logic.

In propositional inclusion logic \cite{yang_propositional_2022}, inclusion atoms with Boolean constants suffice to achieve the desired expressive completeness result. The modal logic setting presents a similar situation, where standard inclusion atoms alone are insufficient for the desired expressivity (see \cite{hella2015,MIL}). Even more extended inclusion atoms are needed, which allow modal logic formulas to function as ``variables'' within the atom, such as $p\subseteq \Diamond q$. A natural extension of this work would be to axiomatize and determine the computational complexity of the implication problems for more subclasses of these extended inclusion atoms.

\begin{credits}
\subsubsection{\ackname}
The author is grateful for helpful suggestions from Åsa Hirvonen, Fan Yang, Miika Hannula and Juha Kontinen. The connection between the axiom schema $B3$ and the complexity class co-NP was Miika Hannula's observation.

\subsubsection{\discintname}
The author has no competing interests to declare that are
relevant to the content of this article.
\end{credits}

\bibliographystyle{splncs04}
\bibliography{mybibliography}

\section*{Appendix}

We give the full proof of \Cref{nonfiniteaxiom}, showing that there is no complete $k$-ary system for inclusion atoms with Boolean constants due to the instance (\ref{rule}) of the axiom schema $B3$.

First, we prove that it is not possible for a set to semantically entail an inclusion atom with $q$ on its right-hand side if no atom in the assumption set contains all of $q$ on its right-hand side.

\begin{lemma}\label{all of q}
    Let $\Sigma$ be a set of inclusion atoms with Boolean constants such that 
\begin{enumerate}[label=$\bullet$]
    \item $\Sigma$ does not contain contradictory atoms, and
    \item  all $u\subseteq v\in \Sigma$ are such that $V(u)\cap V(q)=\emptyset$ and $V(q)\not\subseteq V(v)$.
\end{enumerate}
    Then $\Sigma\not\models u\subseteq q$ for non-trivial consequences $u\subseteq q$.  
    In particular, $\{p^\prime\subseteq q^\prime\mid p^\prime, q^\prime\text{ proper subsequences of }p,q\}\not\models p\subseteq q$.
    %
\end{lemma}

\begin{proof}
If all sequences on the right-hand sides of the atoms are repetition-free and without constants, then build the counterexample team for $u\subseteq q$ according to the proof of \Cref{complincl} using conditions \ref{cond 1} and \ref{cond 3}, observing that now the only sequence not obtaining the value $0^{|q|}$ is $q$. Else, use the counterexample team in the proof of \Cref{CompletenessBoolean} and argue analogously. 
\qed\end{proof}

We can now provide the full proof of \Cref{nonfiniteaxiom}. 

\medskip
 \noindent\textbf{\Cref{nonfiniteaxiom}}  \emph{There is no complete $k$-ary axiomatization of inclusion atoms with Boolean constants.}
\begin{proof}
We consider the rule (\ref{rule}) with the conclusion $p\subseteq q$, $|p|=n$, assuming $V(u)\cap V(q)=\emptyset$ and that $p$ and $q$ are repetition-free sequences of propositional symbols. First, we note that one can easily build teams showing that the whole assumption set (and the conclusion of the rule) need not be satisfied in a team in which only a proper subset of the assumption set is satisfied. Thus, a possible reduction of the arity of the rule is equivalent to the existence of a nontrivial atom derivable from the whole assumption set that is not the conclusion of the rule. We show that this is impossible. In the context of this proof, a trivial consequence would be one derivable using only the unary rules from the system. 

Let $u\subseteq v\neq p\subseteq q$. We only consider cases where $v=q$ due to the following arguments. By the complete proof system in \Cref{rulesBoolean}, any non-trivial transformation of the assumptions would have to involve either $I2$ or $B3$. The rule $I2$ is not applicable to our assumption set, and using $B3$ on strict subset of the assumption set results in atoms of the form $p^\prime\subseteq q^\prime$, with $p^\prime, q^\prime$ being proper subsequences of $p,q$, and by \Cref{all of q} such a collection would be strictly less expressive than the original assumption set. Therefore, we assume that $v=q$ and show that for all such $u\subseteq q$, we can build teams that satisfy the assumptions in the rule (\ref{rule}) but not $u\subseteq q$. We consider all possible variable configurations of $u$ one by one while excluding the previous cases, with the general strategy being to construct teams with constant values for the variables not in $q$, and then produce the minimal team satisfying the assumptions in the rule (\ref{rule}).

We first consider the cases when there is some $r_j$ in $u$, where $r_j$ is a fresh propositional variable not appearing in the sequences $p$ and $q$. This shows that introducing fresh variables cannot help reduce the arity of the rule.\footnote{This is in contrast to the system for inclusion and functional dependencies combined, which is usually considered to not have a $k$-ary axiomatization \cite{CASANOVA198429}, but by allowing the introduction of fresh variables, a $k$-ary system is obtained in  \cite{mitchell}.}

\begin{enumerate}[resume,label=(\alph*)]
\item \label{item aa}
If all variables in $u$ are $r_j$, then build a team where all assignments give $r_j$ the value $1$ and $p$ the value $0^n$. Now the team in \Cref{team r const} shows the claim, by satisfying all assumptions in the rule but not $u\subseteq q$, since $u$ obtains the value $1^n$ while $q$ does not.

  \begin{table}[h]
\centering \caption{} \label{team r const}
$\begin{array}{c cccc cccc}
\toprule	
r_j&p_1 & p_2 & \dots & p_n &q_1& q_2 & \dots &q_n \\ 
 \midrule
1&0&0 & \dots &0&     1 &0& \dots&0 \\ 
1&0&0 &\dots &0&     0 &1& \dots &0 \\ 
& & &\ddots&      && &\ddots  \\ 
1&0&0 & \dots &0&     0 &0& \dots& 1 \\ 
1&0&0 &\dots &0&     0 &0& \dots &0 \\ 
\bottomrule 
\end{array}$ 
\end{table} 

\item 
If there are variables $r_i$ and $r_j$ in $u$ with $i\neq j$, then we build a team in which the value for $r_j$ is $0$ and the remaining variables in $r$ and $p$ are $1$. Now $u$ will have a value with both $0$ and $1$, which is never the case for $q$. This case is illustrated in \Cref{team r}.

  \begin{table}[h]
\centering \caption{}\label{team r}
$\begin{array}{ccc cc cccc cccc}
\toprule	
r_1 & \dots & r_j& \dots&  r_m & p_1 & p_2 & \dots & p_n &q_1& q_2 & \dots &q_n \\ 
 \midrule
1& \dots & 0& \dots&  1 & 1&1 & \dots &1& 1 &1& \dots &1 \\ 
1 & \dots & 0& \dots&  1 & 1&1 & \dots &1& 0 &0& \dots &0 \\ 
\bottomrule 
\end{array}$ 
\end{table}

\item If there is only one $r_j$ and at least one variable from $\{\top, p_i, q_l\}$ in $u$, then we set $r_j$ to be $0$ and build the same team as in the previous case (\Cref{team r}). 

\item The last case is if there is only one $r_i$ and at least one $\bot$ among the variables in $u$. We set $r_i$ to $1$ and $p$ to $1^n$. The team in \Cref{team r} witnesses the claim for the case $r_i=r_1$, since $u$ will obtain a value with both $0$ and $1$, which is never the case for $q$.
\end{enumerate}

Now we can assume that $u_i\in\{p_1,\dots,p_n,q_1,\dots,q_n,\top,\bot\}$ for all $1\leq i\leq n$.

First, let us consider the cases when there is some $u_i$ such that $u_i=\bot$.

\begin{enumerate}[resume,label=(\alph*)]
    \item If there is between $1$ and $n-1$ many $\bot$ in $u$, then we can build the team illustrated in \Cref{team a} (with $p$ set to $1^n$), where $u$ necessarily has a value with both $0$ and $1$, while that is never the case for $q$. \label{item_a} 
    \begin{table}[h]
\centering \caption{}\label{team a}
$\begin{array}{cccc cccc}
\toprule	
p_1 & p_2 & \dots & p_n &q_1& q_2 & \dots &q_n \\ 
 \midrule
1&1 & \dots &1& 1 &1& \dots &1 \\ 
1&1 & \dots &1& 0 &0& \dots &0 \\ 
\bottomrule 
\end{array}$ 
\end{table}
\end{enumerate}

Now assume that $u_i\in\{p_1,\dots,p_n,q_1,\dots,q_n,\top\}$ for all $1\leq i\leq n$.

\begin{enumerate}[resume,label=(\alph*)]
    \item If $u$ contains some variable from $q$ and at least one $\top$, or variables from both $q$ and $p$, then the team in \Cref{team a}  shows the claim, using an analogous argument to \cref{item_a}. 

    \item \label{item_c}   
    If $u_i\in\{p_1,\dots,p_n,\top\}$ for all $1\leq i\leq n$, and there is between $2$ and $n$ many $\top$ in $u$, then we can build \Cref{team c}, where $q$ obtains at most one $1$ on any given line. 
  \begin{table}[h]
\centering \caption{} \label{team c}
$\begin{array}{cccc cccc}
\toprule	
p_1 & p_2 & \dots & p_n &q_1& q_2 & \dots &q_n \\ 
 \midrule
0&0 & \dots &0&     1 &0& \dots&0 \\ 
0&0 &\dots &0&     0 &1& \dots &0 \\ 
& &\ddots &&      && \ddots & \\ 
0&0 & \dots &0&     0 &0& \dots& 1 \\ 
0&0 &\dots &0&     0 &0& \dots &0 \\ 
\bottomrule 
\end{array}$ 
\end{table}
\end{enumerate}

Now the remaining cases are when $u_i\in\{p_1,\dots,p_n,\top\}$ for all $1\leq i\leq n$ with at most one $\top$ in $u$, or $u_i\in\{q_1,\dots,q_n\}$ for all $1\leq i\leq n$. Let us first check the cases for $u_i\in\{q_1,\dots,q_n\}$ for all $1\leq i\leq n$.

\begin{enumerate}[resume,label=(\alph*)]
\item If some variable $q_j$ appears at least twice in $u$, then we can use the team in \Cref{team c}, now $u$ will have a value with two $1$'s, which $q$ never obtains. \label{item h}

\item If there are no repetitions of variables in $u$, then there is a first instance of some variable $q_j$ in the wrong position, lest $u\subseteq q$ is trivial. We set the corresponding $p_j$ to $1$, and the remaining $p_i$'s to $0$. Then $u$ will obtain a value with exactly one $1$, but such a value only occurs for $q$ once, but with $1$ in a different position. In the case of the first instance of a variable in the wrong position being $q_2$, \Cref{team d} shows the claim.\label{item i}
\end{enumerate}

The final cases are of the form $u_i\in\{p_1,\dots,p_n,\top\}$ for all $1\leq i\leq n$.

\begin{enumerate}[resume,label=(\alph*)]
\item If all variables in $u$ are $p_j$, then build a team where $p_j$ is $1$ and the other $p_i$'s are $0$. We give an example in \Cref{team d} when $p_j=p_2$, for which $u$ has the value $1^n$ but $q$ does not. 
\vspace{-.2cm}
\begin{table}[h] \centering \caption{} \label{team d} $\begin{array}{ccccc ccccc} \toprule p_1 & p_2 & p_3 & \dots & p_n & q_1 & q_2 & q_3 & \dots & q_n \\ \midrule 0 & 1 & 0 & \dots & 0 & 1 & 1 & 0 & \dots & 0 \\ 
0 & 1 & 0 & \dots & 0 & 0 & 1 & 0 & \dots & 0 \\ 
0 & 1 & 0 & \dots & 0 & 0 & 1 & 1 & \dots & 0 \\ 
& & & \ddots & & & & &  \ddots& \\ 0 & 1 & 0 & \dots & 0 & 0 & 1 & 0 & \dots & 1 \\ 0 & 1 & 0 & \dots & 0 & 0 & 0 & 0 & \dots & 0 \\ \bottomrule \end{array}$ \end{table}

\item  \label{item e} If some $p_j$ appears at least twice in $u$, set its value to $0$, and the remaining $p_i$'s to $1$. This ensures that each value of $u$ has two $0$'s, which is never the case for $q$. The team in \Cref{team e} proves the claim when $p_j=p_1$.
 \begin{table}[h]
\centering\caption{} \label{team e}
$\begin{array}{cccc cccc}
\toprule	
p_1 & p_2 & \dots & p_n &q_1& q_2 & \dots &q_n \\ 
 \midrule
0&1 & \dots &1&     1 &1& \dots&1 \\ 
0&1 &\dots &1&     0 &1& \dots &1 \\ 
0&1 & \dots &1&     0 &0& \dots& 0 \\ 
\bottomrule 
\end{array}$ 
\end{table}
\item If there is no repetition of any $p_i$ in $u$, then there must be some $p_j$ that occurs in the wrong position, let it be $0$ and the remaining variables in $p$ be $1$. Now for $q$, the corresponding $q_j$ will never be $0$ (except for the line where $q$ has the value $0^n$). The team in \Cref{team e} proves the claim for $p_j=p_1$.
\end{enumerate}\vspace{-.2cm}
\qed\end{proof}

%
%

\end{document}